\documentclass[conference]{IEEEtran}

\usepackage{cite}

\ifCLASSINFOpdf
   \usepackage[pdftex]{graphicx}
   \graphicspath{{Figures/}}
\else
\fi

\usepackage{amsmath}

\usepackage{algorithmic}

\usepackage{array}

\usepackage{fixltx2e}

\usepackage{url}

\usepackage{multirow}
\usepackage{makecell}

\usepackage{booktabs}

\usepackage{tikz}
\usetikzlibrary{shapes,arrows}

\usepackage{amssymb}
\usepackage{amsthm}

\begin{document}

\title{I(FIB)F: Iterated Bloom Filters \\for Routing in Named Data Networks}


%
\author{\IEEEauthorblockN{Cristina Mu\~noz\IEEEauthorrefmark{1},
Liang Wang\IEEEauthorrefmark{1},
Eduardo Solana\IEEEauthorrefmark{4} and
Jon Crowcroft\IEEEauthorrefmark{1}}
\IEEEauthorblockA{\IEEEauthorrefmark{1}
University of Cambridge, United Kingdom \IEEEauthorrefmark{4}University of Geneva, Switzerland\\
\IEEEauthorrefmark{1}\{Cristina.Munoz, Liang.Wang, Jon.Crowcroft\}@cl.cam.ac.uk, \IEEEauthorrefmark{4}Eduardo.Solana@unige.ch
}
}


\maketitle

\begin{abstract}
Named Data Networks provide a clean-slate redesign of the Future Internet for efficient content distribution. Because Internet of Things are expected to compose a significant part of Future Internet, most content will be managed by constrained devices. Such devices are often equipped with limited CPU, memory, bandwidth, and energy supply. However, the current Named Data Networks design neglects the specific requirements of Internet of Things scenarios and many data structures need to be further optimised. The purpose of this research is to provide an efficient strategy to route in Named Data Networks by constructing a Forwarding Information Base using Iterated Bloom Filters defined as I(FIB)F. We propose the use of content names based on iterative hashes. This strategy leads to reduce the overhead of packets. Moreover, the memory and the complexity required in the forwarding strategy are lower than in current solutions. We compare our proposal with solutions based on hierarchical names and Standard Bloom Filters. We show how to further optimise I(FIB)F by exploiting the structure information contained in hierarchical content names. Finally, two strategies may be followed to reduce: (i) the overall memory for routing or (ii) the probability of false positives.
\end{abstract}


%
\IEEEpeerreviewmaketitle


\section{Introduction}
\label{sec:introduction}

Nowadays, Internet works mainly as a distribution network. The introduction in the market of new devices such as smartphones, tablets, wearables, sensor nodes, home appliances, etc. that compose the Internet of Things (IoT) leads to ever-growing global content. Moreover, the new ways of communication that include e-commerce, social and media networks also require data dissemination. Therefore, a future architecture to address the challenges confronting our current Internet has attracted a lot of research interest in the networking community. Information-Centric Networking (ICN) \cite{kutscher2016information} is proposed based on two fundamental concepts: i) accessing content by name; ii) universal caching. With both, the paradigm is shifted from point-to-point communication to information-centric dissemination.

Named Data Networking (NDN) \cite{zhang2014named} constitutes one of the most promising architectures for content management. The main advantage of having data as focal point is that no point-to-point connection is required. As a consequence, the flexibility and efficiency of the network is improved by the ``democratization of the Internet'' where connected entities may publish or request data based on content. Additionally, NDN can improve the latency by pushing content even closer to clients comparing to Content Delivery Networks (CDN). Moreover, NDN is more capable of capturing more network dynamics (temporal and spatial locality regarding users' request streams).

Finally, NDN provides inherent security because data is encrypted itself and digitally signed instead of partially protected only while in transit between specific end-points. Some of the benefits being that content intellectual property is preserved or that bandwidth is optimized through automatic caching.

Our research focuses on the key data structure Forwarding Information Base (FIB) and the forwarding strategy. We propose I(FIB)F, a novel FIB design based on Iterated Bloom Filters (IBFs). The efficiency of our solution in terms of complexity and memory requirements makes it very suitable for constrained devices. The results of our investigations confirm that IBFs may reduce the probability of false positives or the memory-positions required for routing. Furthermore, we show that content names may be estimated field by field to obtain more precise measurements required to design I(FIB)Fs.

Specifically, our contributions are as follows:

$\bullet$ We propose to use I(FIB)F to replace the forward Interest table in the current NDN design, along with detailed performance analysis.

$\bullet$  Comparing to the original FIB using hierarchical naming, our analysis shows that I(FIB)F can significantly reduce the traffic overhead, storage overhead, as well as computation overhead in forwarding.

$\bullet$  The comparison of I(FIB)F with Standard BFs concludes that the overall memory for routing and the probability of false positives may be reduced.

$\bullet$ We optimise I(FIB)F by exploiting the structure information contained in hierarchical content names.

$\bullet$  We present how to adopt the proposed solution in the existing NDN protocol stack and show the induced engineering overhead is minimal.

The rest of this paper is organized as follows: Section \ref{sec:iterated} details the background that motivates this work. Section \ref{sec:system} describes the design of the system. Section \ref{sec:analysis} analyses our design. Section \ref{sec:routing} presents how to adapt I(FIB)F to a current routing solution. Section \ref{sec:related} points out related work. Finally, Section \ref{sec:conclusion} summarizes our proposal.

\section{Iterated Bloom Filters}
\label{sec:iterated}

For the sake of clarity a notation table is given at Table \ref{table:notation}.

\begin{table}[!h]
\centering
\scriptsize
\caption{Notation table.}
\label{table:notation}
    \begin{tabular}{ll}\toprule  
    
    $b$ & Number of bits to define a memory-position of Standard BFs\\ 
    
    $b_c$ & Average number of bits per character of a hierarchical structure\\ 
	
	$b_{ind_x}$ & Number of bits to define a memory-position of an IBFs of level $x$\\ 
	
	$c$ & Average number of characters on a hierarchical structure\\ 
	
	$d$ & Number of levels of IBFs\\ 

	$f$ & Probability of false positives for Standard BFs\\
	
	$f_{i}$ & Overall probability of false positives of IBFs\\  
	
	$f_{ind_x}$ & Probability of false positives of an IBF of level $x$\\  
	
	$k$ & Overall number of hashes of Standard and IBFs\\ 
	
	$k_{i}$ & Number of hash functions of IBFs\\
	
	$k_{ind_x}$ & Number of hashes of an IBF of level $x$\\ 

	$m$ & Overall memory-positions of Standard BFs\\ 
	
	$m_{i}$ & Memory-positions of all IBFs\\ 
	
	$m_{ind_x}$ & Memory-positions of an IIBF of level $x$\\ 
	
	$n$ & Number of elements to insert to a Standard BF\\ 
	
	$n_{i}$ & Number of elements to insert to IBFs\\ 
	
	$n_{ind_x}$ & Number of elements to insert to an IBF of level $x$\\ 
	
	$p$ & Probability that a bit is still 0 after inserting all $n$ in a Standard BF\\ 
	
	$p_{i}$ & Probability that a bit is still 0 after inserting all $n_{i}$ in IBFs\\
	
	$p_{ind_x}$ & Prob. that a bit is still 0 after inserting all $n_{ind_x}$ in an IBF of level $x$\\

	$\mu$ & Average number of elements expected\\

	$\sigma$ & Standard deviation of the elements expected\\
	
	$\sigma^2$ & Variance of the elements expected\\
 
  \bottomrule
    \end{tabular}

\end{table}

Iterated or Merkle-Damg\r{a}rd hash functions  \cite{katz2014introduction} hash an input iteratively by feeding the output of each iteration into the input of the next. This design is very convenient when dealing with tree structures because the previous hashes contribute to obtain new branches. Then, a tree that saves the hashes of the following structure: $Field1/Field2$, has on top the hash of $Field 1$, $h_1$ which is computed as $h(Field1)$. The following hash $h_2$ may be computed as (i) $h(Field1 || Field2)$ or as (ii) $h(h_1 || Field2)$, where $||$ represents a concatenation. We note that by computing $h_2$ using the second strategy computing resources are saved because we profit from the previous computed hash.

Iterative Bloom Filters (IBFs) \cite{munoz2015fragmented} take advantage of the properties of iterative hash functions. The strategy followed by IBFs is to split the $m$ bit-positions of a Standard BF to save the same number of elements $n$. Then, a Standard BF \cite{Standard} may be split in $d$ IBFs of $m/d$ bit-positions and $n$ elements.

Table \ref{table:standard vs iterated fixing p} from \cite{ThesisCFIBFs} shows that when fixing the probability that a bit is still 0 after inserting all elements in Standard BFs and IBFs, the memory size, the number of elements and the probability of false positives is maintained. The advantage of using IBFs is that the computation is reduced because (i) they require less hashes and (ii) Individual BFs benefit from the properties of iterative trees.

\begin{table}[!ht]
\centering
\scriptsize
\caption{Standard BF vs IBFs: fixing the probability that a bit is 0 after inserting all elements according to the Standard BF.}
\label{table:standard vs iterated fixing p}
    \begin{tabular}{rcrcr}\toprule  

\textsc{Standard BF} & & \textsc{Individual IBF} & & \textsc{IBFs}\\
\midrule

    $m$ & & $m_{ind}=m/d$& & $m_i=m_{ind} \cdot d=m$ \\
    
    $n$ & & $n_{ind}=n$ & & $n_i=n$ \\ 
    
    $p=e^{-kn/m}$ & & $p_{ind}=p$& & $p_{i}=p$ \\ 
    
    $k=\frac{-m}{n}ln p  $ & & $k_{ind}= \frac{-m/d}{n}\ln p=\frac{k}{d}$ & & $\textcolor{blue}{k_i = k_{ind}=\frac{k}{d}}$ \\  
                
    $f=(1-p)^{k}$ & & $f_{ind}=(1-p)^{k/d}$&  & $f_{i}=f_{ind}^d=f$ \\ 

  \bottomrule
    \end{tabular}

\end{table}

\section{System Design}
\label{sec:system}


\subsection{Network model}
\label{sec:network}

Our network model is based on an IoT scenario wherein the connected devices such as wearables and wireless sensor nodes are constrained by their available CPU, memory, bandwidth, and battery resources.

In IoT context, the network is often formed in an ad-hoc way due to high mobility of devices and there is no fixed infrastructure, which, can further justify the need of low cost routing to deal with the induced churn. For this reason, we consider that the network deployment should be flexible enough and each node is in charge of measuring traffic to estimate the $Interests$ expected. This estimation is necessary to design the I(FIB)F according to the number of different $Names$ that may be routed through a node.


\subsection{Naming}
\label{sec:naming}

We consider a flat naming scheme to identify $Interests$. I(FIB)F enables the use a flat naming schemes without sacrificing the benefits of hierarchical structures in content names. Therefore, unlike other flat naming scheme, our proposed scheme can also achieve equivalent name aggregation as that in original NDN FIB.  

For the purpose of clarification we show how to hash a hierarchical name with an example. Let us consider that a new $Interest$ arrives with the $Name$: $Cambridge/ComputerLab/FW01/Windows$.

Figure \ref{fig:IBFs1} shows how to save the element using a Standard BF that requires four hashes per element and four Iterated BFs that require one hash each one. 

The upper part of Figure \ref{fig:IBFs1} shows how the $Name$ is embedded in a Standard BF (of 32 bit-positions) using four different hash functions $h1$, $h2$, $h3$, and $h4$. The four hash functions and their corresponding output positions are marked with different colors.

We show then how the same $Name$ is embedded in four IBFs (of 8 bit-positions each one). A single hash function is required. The iterated hash outputs for each field feed Individual IBFs. 

Then, in our solution we send the iterated hashes of each field using the hash function $h1$ so that we substitute the hierarchical name by: $h1a(Cambridge)$, $h1b( h1a || /ComputerLab)$, $h1c( h1b || /FW01)$ and $h1d( h1c || /Windows)$. We observe that with only one hash of the input and benefiting from the iterative outputs the element is saved whereas a Standard BF computes four separate hashes using the whole content name.

In case that the design of the IBFs requires two different hash functions we will combine also the outputs produced by the hash function $h2$: $h2a(Cambridge)$, $h2b( h2a || /ComputerLab)$, $h2c( h2b || /FW01)$ and $h2d( h2c || /Windows)$. We repeat this operation for the specified number of hash functions required by the system. Section \ref{sec:configuring} discusses how to adjust the total number of hash functions in a particular node.

\begin{figure}[tp!]
\centering{
\begin{tabular}{cc}
 	\includegraphics[width=1\linewidth]{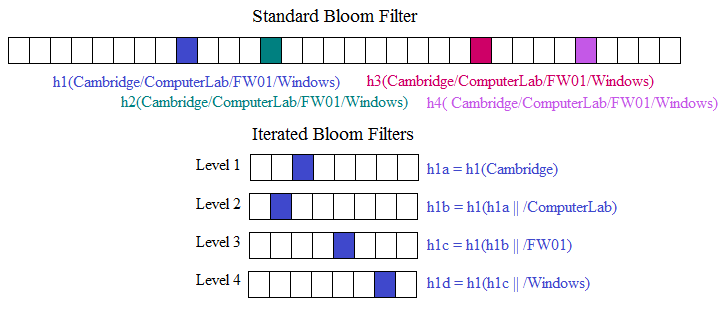}
\end{tabular}}
\caption{Hashing the element: \textit{Cambridge/ComputerLab/FW01/Windows}. The Standard BF requires four hash functions whereas IBFs require only one hash function. Colors represent the outputs of the same hash function.}
\label{fig:IBFs1}
\end{figure}

One of the properties of IBFs is that if top leaves coincide, the probability of false positives is reduced. This is due to the fact that some elements provide the same hash at some IBFs. Figure \ref{fig:IBFs2} follows the example provided in Figure \ref{fig:IBFs1} to show the final state of BFs when hashing four elements with similarities in the first fields of the structure. We observe that the IBFs that save the first fields in the structure have less bits set to 1 due to the fact that some elements coincide. This leads to the reduction in the probability of false positives at IBFs.

\begin{figure}[tp!]
\centering{
\begin{tabular}{cc}
 	\includegraphics[width=1\linewidth]{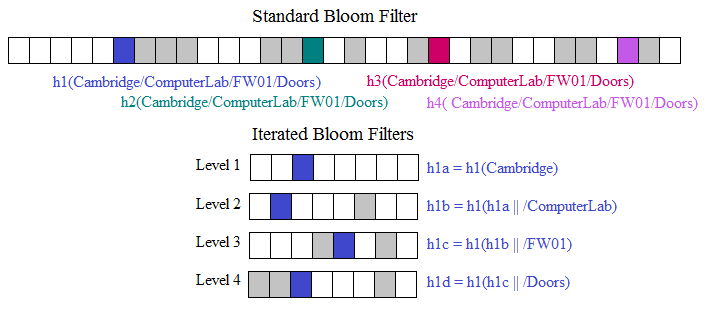}
\end{tabular}}
\caption{Hashing the elements: \textit{Cambridge/ComputerLab/FW01/Windows, Cambridge/ComputerLab/FW26/Windows, Cambridge/Physics/Mott/Motion, Cambridge/ComputerLab/FW01/Doors.} The grey positions represent the first three elements. IBFs finalise with more free positions, this leads (i) to use less memory at top levels or (ii) to decrease the probability of false positives.}
\label{fig:IBFs2}
\end{figure}

It is worth pointing out that the depth of hierarchy in content names (i.e., level $d$) is not limited by our scheme. In our example, when an element of more than four fields is received it is saved as shown from level one to three. Afterwards, the last fields are iteratively hashed and the last hash output is used to fill in the Individual IBF of level four.

\subsection{Forwarding Information Base}
\label{sec:forwarding}

\begin{figure}[!ht]
\centering{
\begin{tabular}{c}
 	\includegraphics[width=0.7\linewidth]{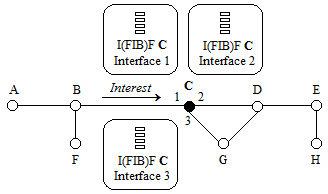}
\end{tabular}}
\caption{Node C receives an $Interest$ through interface 1 and checks their I(FIB)Fs at interfaces 2 and 3 for routing.}
\label{fig:Network}
\end{figure}


In NDN, there is only one original FIB per node. This FIB maps each $Name$ with the corresponding next hop(s), so that when an $Interest$ arrives the appropriate output interface is selected for forwarding the packet. In our system, we require to associate each interface with a separate I(FIB)F composed of IBFs. After, when an $Interest$ arrives, the iterated hashes contained in the $Name$ are directly checked against each possible output interface. Figure \ref{fig:Network} shows that node C has three I(FIB)Fs, one per interface.

The forwarding strategy requires to implement a membership test at IBFs of each FIB. Previously, a routing protocol needs to fill in IBFs according to its policy. In section \ref{sec:routing}, we propose to use the well-known Named-data Link State Routing Protocol (NLSR) \cite{hoque2013nlsr} to compute the shortest path to all nodes in the network.

When an $Interest$ arrives, the iterated hashes of the $Name$ are used to check the appropriated positions of IBFs associated with possible output interfaces. In Figure \ref{fig:Network} when an $Interest$ arrives through interface 1, the FIBs of Interfaces 2 and 3 are checked. Afterwards, the interface with the largest match is selected for forwarding which is similar to the largest prefix matching in the original FIB design but with much less computation overhead. Following with our example of Figure \ref{fig:IBFs2}, whether interface 2 provides a match for $h1a(Cambridge)$ and interface 3 provides a match for $h1a(Cambridge)$, $h1b( h1a || /ComputerLab)$ and $h1c( h1b || /NetOS)$. Then, interface 3 is selected for forwarding.


\subsection{Configuring I(FIB)F}
\label{sec:configuring}

Before designing each I(FIB)F, it is necessary to estimate the number of elements which may be received. We study two different methods to estimate elements in Section \ref{sec:fibanalysis}. We optimize our design by exploiting the structure contained in hierarchical names using a well-known estimation method as Bayesian statistics \cite{bolstad2013introduction}.

Once we have estimated the total number of levels and the elements at each level, we can design the I(FIB)F for a certain confidence interval. Typically, normal distributions are designed to offer a 95\% of coverage. This implies that the elements expected should be computed as $\mu + 1.96·\sigma$. Please, refer to Table \ref{table:notation} for notations. Other confidence intervals may be selected. For instance, if we choose to cover the 68\%, 90\% and 99\% of elements we may compute the elements expected as $\mu + \sigma$, $\mu + 1.65·\sigma$ or $\mu + 2.58·\sigma$. The estimation of a larger $\sigma$ leads to overestimate the memory requirement reducing the efficiency of our solution.

Given an estimated $n_{ind_x}$ for a coverage of 95\%, the I(FIB)F may be designed for $n_{ind_x}=\mu + 1.96·\sigma$ and a fixed $p_{ind_x}$.

Herein we consider two cases: i) non-repetition: there is no overlap in all levels between any two content names. Then,  I(FIB)F may be designed for a certain $f_{ind_x}$ that fixes $k_{ind_x}$ and $m_{ind_x}$; ii) repetitions: we allow different content names to share common entities in various levels as previously discussed in Section \ref{sec:iterated}. In this case, two strategies may be followed:

\textit{Strategy I:} Memory-positions may be saved by maintaining the overall probability of false positives. In this case, the relation $m_{ind_x}/n_{ind_x}$ must be maintained. As a result, $m_{ind_x}$ is reduced according to the new $n'_{ind_x}$ so that $n'_{ind_x} < n_{ind_x}$ and the final memory-positions may be computed as: 

\begin{equation}
m'_{ind_x}=\frac{m_{ind_x} \cdot n'_{ind_x}}{n_{ind_x}}
\end{equation}

\textit{Strategy II:} The probability of false positives may be reduced by maintaining the memory-positions. In this case, $m_{ind_x}$ is maintained and the new number of hashes $k'_{ind_x}$ may be computed as according to the new $n'_{ind_x}$:

\begin{equation}
k'_{ind_x}=-\frac{m_{ind_x}}{n'_{ind_x}}\ln p_{ind_x}
\end{equation}

Which leads to the following probability of false positives:

\begin{equation}
f'_{ind_x}\approx(1-p_{ind_x})^{k'_{ind_x}}
\end{equation}

As detailed in Section \ref{sec:iterated}, if $f_{ind_x}$ is reduced then the overall probability of false positives for IBFs $f_{i}$ is also reduced.

It is remarkable to mention that if the strategy I is followed, we may have IBFs of different memory-positions while if we consider strategy II all IBFs will have the same number of memory-positions.

Finally, in the event that we receive more elements than expected, we may use Dynamic Bloom Filters (DBF) \cite{guo2010dynamic} to prevent information loss. A DBF is a new BF that starts saving new elements when the old BF has reached the limit of elements to accept. In our case, a new IBF must be added to the existing one for a specific level. Moreover, elements may be removed by adding a counter to each position \cite{bonomi2006improved}.


\section{Analysis}
\label{sec:analysis}


\subsection{Naming}
\label{sec:namingeval}
In this section, we evaluate the performance of three different naming strategies. For the sake of simplicity we use four-level names in the following analysis, but note that the analysis and our conclusions can be equally generalized to names of arbitrary depth/level.

\textbf{Hierarchical structure: } When using a hierarchical structure, the whole construction is required: $Field1/Field2/Field3/Field4$. In the worst case, when identifying each field with a real name, one word is necessary at each field to define the structure. The average word length in English is 4.5 \cite{shannon1951prediction}. We assume that 8 bits are required per character. Following our example, we require 144 bits.

\textbf{Standard BFs: }Standard BFs require $k$ different hashes. Then, when using this strategy $k$ hashes $h(Field1/Field2/Field3/Field4)$ are necessary.

\textbf{Iterated BFs: }IBFs require to individually save each field. Then, $k/d$, in our example $k/4$, different hashes of each field $h(Field1)$, $h(h(Field1)||Field2)$, $h(h(h(Field1)||Field2)||/Field3)$ and $h(h(h(h(Field1)||Field2)||/Field3)||/Field4)$ are necessary. So the total number of hashes required to define IBFs is the same than for Standard BFs. 


When designing an I(FIB)F, one of the parameters to specify is $k_{i}$. This parameter affects the number of bits required for the naming. Then, as the number of needed hashes decreases, the less naming bits are used, so that the overhead in the $Name$ is reduced. Figure \ref{fig:namingbitslevels} shows the total naming bits required depending on the number of hashes for three different memory sizes. Figure \ref{fig:namingbitslevels}.a shows that for a memory of 16.38 kB less than 9, 5 and 3 hashes are required in a Standard BF (SBF), IBFs of 2 levels (2IBFs) and IBFs of 4 levels (4IBFs) to transmit less bits than a hierarchical structure.

The boundary on the number of hashes at Standard BFs is computed as:

\begin{equation}
k \leq \frac{c \cdot b_c}{b}
\end{equation}


Assuming that IBFs have similar memory-positions, so that all $b_{ind_x}$ are equal, then the boundary on the number of hash functions at IBFs is computed as:

\begin{equation}
k_{i} \leq \frac{c \cdot b_c}{d \cdot b_{ind_x} }
\end{equation}


The comparison of Figures \ref{fig:k}.a, \ref{fig:k}.b and \ref{fig:k}.c shows that bigger memories admit a larger number of elements for the same number of hash functions. Moreover, $f$ is the same in all cases (see Figure \ref{fig:f}.a). We note that larger number of hash functions lead to admit less elements whereas $f$ is decreased. Then, we conclude that for a fixed memory the transmission of less naming bits, which implies the use of less hash functions, increases $f$. 

Figure \ref{fig:namingbitslevels}.a shows that SBF, 2IBFs and 4IBFs of four, two and one hash functions require 68, 64 and 60 naming bits for a $f$ of 0.0625 and 22710 elements. The improvement with respect to the hierarchical structure is notable because less than half of the bits are transmitted. IBFs composed of more levels transmit even less bits because its $b_{ind_x}$ is lower.

If we decide to reduce $f$ we increase the number of hash functions. In the case of 4IBFs and one hash function we decrease to the half the number of elements for a probability of 0.0039, which is 16 times better than for one hash function. This change affects also to the naming bits that are doubled.

\begin{figure*}[!ht]
\centering{
\begin{tabular}{ccc}
 	 \includegraphics[width=0.31\linewidth]{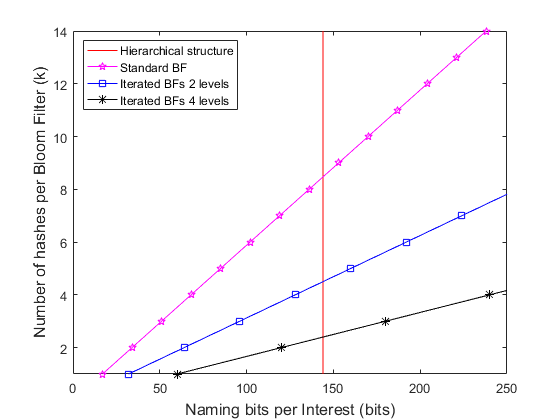}
		\label{fig:namingbitslevelsm1} &
 	\includegraphics[width=0.31\linewidth]{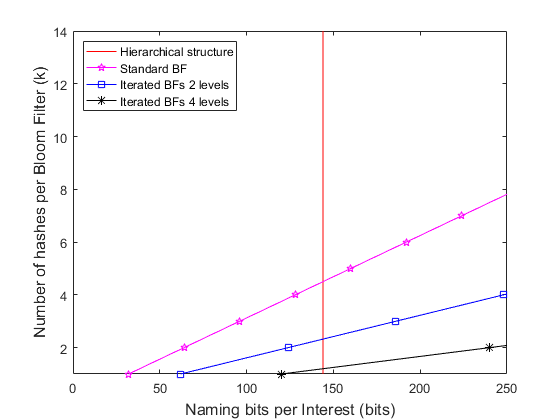}
		\label{fig:namingbitslevelsm2} &
 	\includegraphics[width=0.31\linewidth]{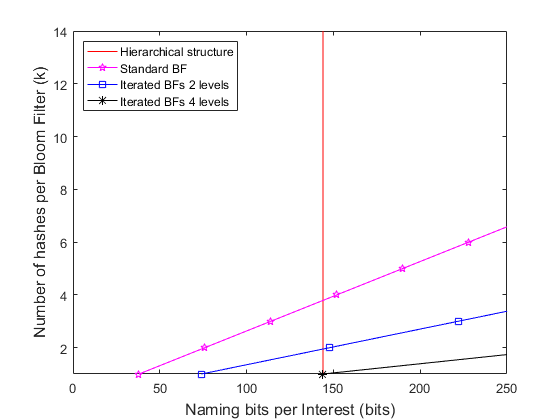}
		\label{fig:namingbitslevelsm3}\\
  	 \scriptsize a) Case I: $m$=16.38 kB, $b$=17, $p$=0.5 & \scriptsize b) Case II: $m$=536.87 MB, $b$=32, $p$=0.5 & \scriptsize c) Case III: $m$=34.36 GB, $b$=38, $p$=0.5\\
\end{tabular}}
\caption{Total naming bits transmitted for different memory sizes. The number of hash functions required is shown for each case.}
\label{fig:namingbitslevels}
\end{figure*}

\begin{figure*}[!ht]
\centering{
\begin{tabular}{ccc}
 	 \includegraphics[width=0.31\linewidth]{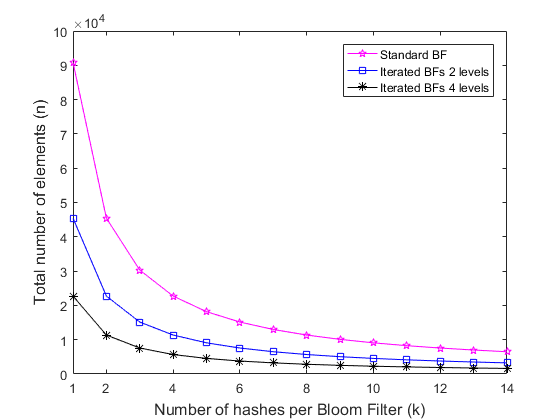}
		\label{fig:km1} &
 	\includegraphics[width=0.31\linewidth]{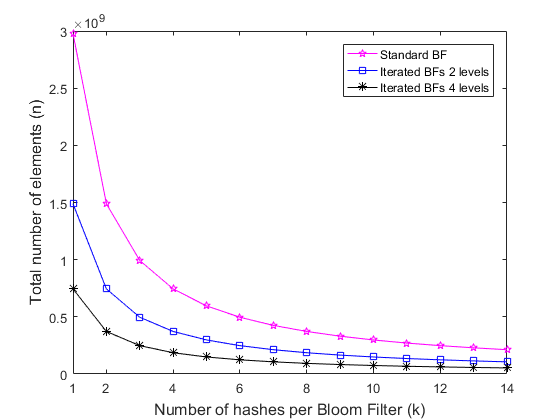}
		\label{fig:km2} &
 	\includegraphics[width=0.31\linewidth]{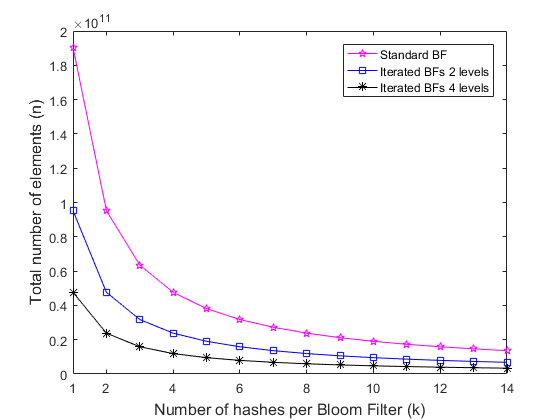}
		\label{fig:km3}\\
  	 \scriptsize a) Case I: $m$=16.38 kB, $b$=17, $p$=0.5 & \scriptsize b) Case II: $m$=536.87 MB, $b$=32, $p$=0.5 & \scriptsize c) Case III: $m$=34.36 GB, $b$=38, $p$=0.5\\
\end{tabular}}
\caption{Number of elements $n$ admitted depending on the number of hash functions for different memory sizes.}
\label{fig:k}
\end{figure*}

As previously discussed, $f$ decreases if repetitions occur (see Figure \ref{fig:f}.b). Table \ref{table:repetitions} shows this effect for SBFs, 2IBFs and 4IBFs of 16.38 kB designed for four, two, and one hash functions. 

In case of no-repetitions $f$ is 0.0625. We assume that the same name structures are received. In the first example, SBF has 5\% of repeated elements. To keep consistency, the last level of the IBFs also admits 5\% of repetitions. 2IBFs admit 20\% of repetitions at the first level. 4IBFs admit 50\%, 20\% and 10\% of repetitions from the first to the third level. Under these conditions, 4IBFs reduce $f$ by one half with respect to SBFs. We conclude that IBFs of larger levels benefit from repetitions.

In the second example, all percentages are increased by 10\%. We observe that all $f$s are reduced and that the difference between SBFs and 4IBFs has increased when compared with the previous example. 

Finally, we consider repetitions of 50\% at all levels. As a result, all $f$s are equal. We conclude that IBFs benefit from repetitions if different percentage of repetitions are expected at each level.

\begin{figure*}[!ht]
\centering{
\begin{tabular}{ccc}
 	 \includegraphics[width=0.31\linewidth]{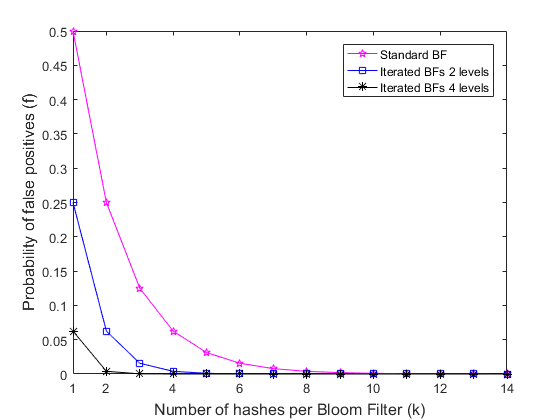}
		\label{fig:f_b17} &
 	\includegraphics[width=0.31\linewidth]{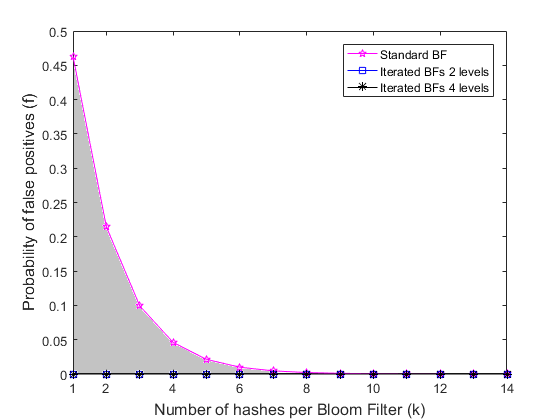}
		\label{fig:repe5025105m1} &
	\includegraphics[width=0.31\linewidth]{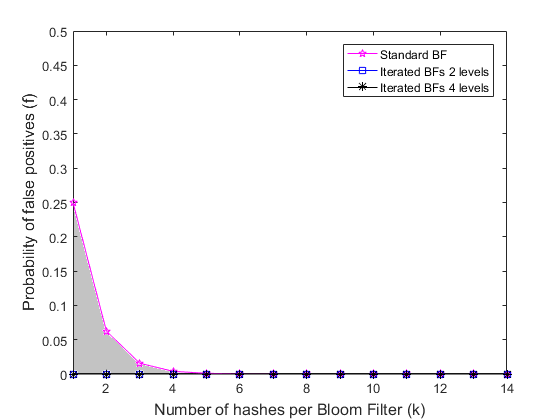}
		\label{fig:repe5025105m1} 	
		\\
  	 \scriptsize a) No repetitions & \scriptsize b) Last BF fixed to 10\% of repetitions & \scriptsize c) Last BF fixed to 50\% of repetitions\\
\end{tabular}}
\caption{Probability of false positives when (a) no-repetitions occur. Range for the probability of false positives in case of 2IBFs and 4IBFs when: (b) the last level is fixed for repetitions of 10\% and (c) the last level is fixed for repetitions of 50\%.}
\label{fig:f}
\end{figure*}

\begin{table}[!h]
\centering
\scriptsize
\caption{Overall and Individual probabilities of false positives for different percentage of element repetitions.}
\label{table:repetitions}
    \begin{tabular}{rcrrcrrcrr}\toprule  

& \multicolumn{2}{r}{\textsc{Example I}} && \multicolumn{2}{r}{\textsc{Example II}} && \multicolumn{2}{r}{\textsc{Example III}} \\

& \textsc{$f$} & \textsc{$n_{rep}$} & & \textsc{$f$} & \textsc{$n_{rep}$} & & \textsc{$f$} & \textsc{$n_{rep}$}\\
\midrule

	
	\textsc{Standard BF} & 0.0541 & 5\% & & 0.0393 & 15\% & & 0.0074 & 50\% \\
\midrule
	\textsc{2 level IBFs} & 0.0422 &  & & 0.0261 &  & & 0.0074 & \\ 
\addlinespace
	\textsc{Level 1} & 0.1812 & 20\%  && 0.1316 & 35\%  & & 0.0858 & 50\% \\
	\textsc{Level 2} & 0.2327 & 5\%   && 0.1982 & 15\%  & & 0.0858 & 50\% \\
\midrule	
	\textsc{4 level IBFs} & 0.0266 &  & & 0.0166 &  &  & 0.0074 \\ 
\addlinespace
	\textsc{Level 1} & 0.2929 &  50\%  & & 0.2421 & 60\% & & 0.2929 & 50\% \\
	\textsc{Level 2} & 0.4054 &  25\%  & & 0.3627 & 35\%  & & 0.2929 & 50\% \\
	\textsc{Level 3} & 0.4641 &  10\% & & 0.4257 & 20\%& & 0.2929 & 50\% \\
	\textsc{Level 4} & 0.4824 & 5\% & & 0.4452 & 15\% & & 0.2929 & 50\% \\
   
  \bottomrule
    \end{tabular}

\end{table}

Furthermore, we evaluate $n$ by assuming a fixed memory-size and different FIBs. Figure \ref{fig:fa}.a compares the different strategies for different $f$ and one I(FIB)F. We observe, that all strategies that use BFs require the same $n$ for a fixed $m$. We remark that hierarchical structures do not depend on $f$. As discussed before, if $n$ increases then $f$ increases as well. If decreasing $f$ is a design
requirement, we must also decrease $n$.

We would like to notice that Case III has been designed to obtain the same number of naming bits for hierarchical structures and 4IBFs. In this case, IBFs have an $f$ of 0.0625 while hierarchical structures have no $f$. However, memories of 34.36 GB admit $4763\cdot10^{10}$ elements without repetition in case of 4IBFs and a single I(FIB)F whereas hierarchical structures may hold $1909\cdot10^6$ elements. The benefit of using IBFs is obvious by noticing there are several orders of magnitude difference in memory efficiency.


Finally, Figure \ref{fig:fa}.b shows how $n$ decreases when admitting more I(FIB)Fs in a node, so that the number of interfaces is increased. We observe that even for ten I(FIB)Fs, SBFs, 2IBFs and 4IBFs admit more elements for the same memory-size. In the Future Internet it is likely that limited devices, in terms of CPU, memory, bandwidth and battery, will not accept a large number of interfaces. As a result, we state that our proposal is very convenient and superior for these scenarios.


\begin{figure*}[!ht]
\centering{
\begin{tabular}{cc}
 	 \includegraphics[width=0.31\linewidth]{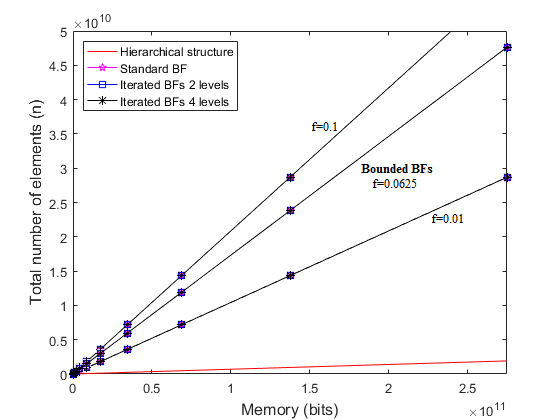}
		\label{fig:adjusted} &
 	\includegraphics[width=0.31\linewidth]{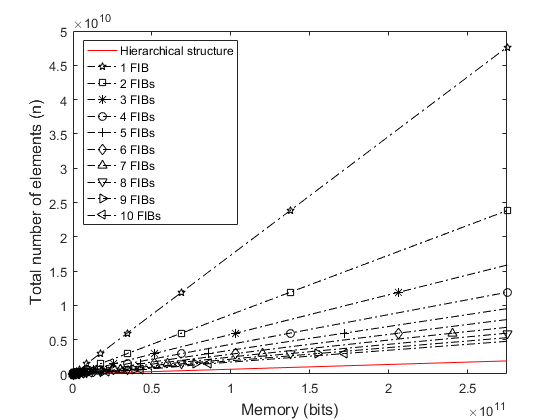}
		\label{fig:f0.01} \\
	\scriptsize a) One FIB for Standard and Iterated BFs of $f=0.0625$ (bounded), $f=0.1$  and $f=0.01$ & \scriptsize b) From one to ten FIBs using bounded Standard and Iterated BFs \\
\end{tabular}}
\caption{Memory required in cases I, II and II for: a) bounded Standard and Iterated BFs of $f=0.0625$ and for Standard and Iterated BFs of $f=0.1$  and $f=0.01$ when having a single FIB; b) bounded Standard and Iterated BFs when having different FIBs.}
\label{fig:fa}
\end{figure*}

\subsection{Estimation of Names}
\label{sec:fibanalysis}


As detailed in Section \ref{sec:forwarding} we design I(FIB)F for a certain confidence interval. Then, a precise estimation of the standard deviation is fundamental to reduce the uncertainty on the estimation of the number of elements per field.

Let us consider a well-known estimation method based on Bayesian statistics \cite{murphy2007conjugate}. To simplify our analysis, we study a conjugate Bayesian of the Gaussian distribution assuming a certain variance. Let $X=(x_1, x_2, ..., x_n)$ be the observed content names in a request stream. It is known that $n$ observations with variance $\sigma^2$ and mean $\bar{x}$ which are normally distributed are equivalent to a likelihood of $\mathcal{N}=(\bar{x}, \frac{\sigma^2}{n})$. Furthermore, we need to define the distribution of data as the prior $\mathcal{N}=(\mu_0, \sigma^2_0)$. Typically, a prior is defined in the absence of data so that $\mu_0=0$. Additionally, we must take into account that if the distribution of data is not well-defined larger $\sigma^2_0$  must be preferred to assume a non informative prior. Under these assumptions, the posterior is defined as:
\begin{equation}
\mathcal{N}=(\mu_n, \sigma^2_n)
\end{equation}
where the mean  and the variance are:
\begin{equation}
\mu_n=\sigma^2_n(\frac{\mu_0}{\sigma^2_0}+\frac{n\bar{x}}{\sigma^2})^{-1}
\end{equation}
\begin{equation}
\sigma^2_n=(\frac{n}{\sigma^2}+\frac{1}{\sigma^2_0})^{-1}
\end{equation}

This solution is suitable for Standard BFs because whole structures are hashed, so that in case of 95\% of coverage we estimate to receive $\mu_n + 1.96·\sigma_n$ elements.

In the case of I(FIB)F, it is necessary that a method estimates the frequency of a variable at a certain field of the structure. We know that if top structure fields coincide, they are hashed to the same position and, as a consequence, the IBFs may be designed for receiving less number of elements. For the sake of simplicity, let us define hierarchical structures of up to two fields that may contain the variables ${a, b}$. Therefore, we need to estimate $a$, $b$, $a/b$ and $b/a$ during a certain period of time from $0$ to $t$. At this point, two different methods may be used:

\textit{Method I:} When a hierarchical structure is received the whole name is estimated. The obtained posteriors are the following ones:

$\mathcal{N}_{a'_1}(\mu_{a'_1}, \sigma^2_{a'_1})$, $\mathcal{N}_{b'_1}(\mu_{b'_1}, \sigma^2_{b'_1})$, $\mathcal{N}_{(a/b)'_1}(\mu_{(a/b)'_1}, \sigma^2_{(a/b)'_1})$ and $\mathcal{N}_{(b/a)'_1}(\mu_{(b/a)'_1}, \sigma^2_{(b/a)'_1})$.
 
Afterwards, when all observations have been taken, we need to add the distributions that contain the same variable in a certain field. In our example, the final estimations are: 

$\mathcal{N}_{a_1}(\mu_{a'_1}+\mu_{(a/b)'_1}, \sigma^2_{a'_1}+\sigma^2_{(a/b)'_1})$, $\mathcal{N}_{b_1}(\mu_{b'_1}+\mu_{(b/a)'_1}, \sigma^2_{b'_1}+\sigma^2_{(b/a)'_1})$, $\mathcal{N}_{(a/b)_1}=\mathcal{N}_{(a/b)'_1}$ and $\mathcal{N}_{(b/a)_1}=\mathcal{N}_{(b/a)'_1}$.



\textit{Method II:} When a hierarchical structure is received, each variable of a field is independently estimated. Following our example, the final estimations are: 

$\mathcal{N}_{a_2}(\mu_{a'_2}, \sigma^2_{a'_2})$, $\mathcal{N}_{b_2}(\mu_{b'_2}, \sigma^2_{b'_2})$, $\mathcal{N}_{(a/b)_2}(\mu_{(a/b)'_2}, \sigma^2_{(a/b)'_2})$ and $\mathcal{N}_{(b/a)_2}(\mu_{(b/a)'_2}, \sigma^2_{(b/a)'_2})$.


\newtheorem{theorem}{Theorem}
\begin{theorem}
\label{variance}
We state that the design of I(FIB)F is more accurate if an estimation method is required field by field of the hierarchical names.
\end{theorem}

\begin{proof}
It is well-known \cite{murphy2007conjugate} that Bayesian methods reduce the variance of estimations when increasing the number of observations. We also know that the addition of Gaussian distributions increase the uncertainty by adding their variances. As a consequence, in our example $\sigma^2_{a_1}=\sigma^2_{a'_1}+\sigma^2_{(a/b)'_1}$, $\sigma^2_{a_2}=\sigma^2_{a'_2}$ ,$\sigma^2_{b_1}=\sigma^2_{b'_1}+\sigma^2_{(b/a)'_1}$ and $\sigma^2_{b_2}=\sigma^2_{b'_2}$; so that $\sigma^2_{a_1}>\sigma^2_{a_2}$ and $\sigma^2_{b_1}>\sigma^2_{b_2}$. We conclude that Method II is more precise than Method I because its uncertainty is smaller.
\end{proof}
Theorem \ref{variance} is important because it provides a more realistic view of the number of elements that should be inserted at each Individual IBF. The disadvantage when comparing our solution with Standard BFs is that we need to check field by field all the structures received, although this is only necessary once. Moreover, traffic measurement \cite{adamic2002zipf} can be calculated off-line. Typically, the statistics remain stable, they will not change in a short time, so the overall overhead is small. In any case, the benefits of IBFs overcome this inconvenience specially when having large number of coincidences at top fields of $Names$.

\subsection{Configuring I(FIB)F}
\label{sec:fibconfiguring}





As previously discussed, we may design an I(FIB)F following two strategies: (i) we may save memory-positions by maintaining the overall probability of false positives (see Theorem \ref{reduce_m_maintain_f}) or (ii) we may reduce the overall probability of false positives by maintaining the overall memory-positions (see Theorem \ref{reduce_f_maintain_m}).

Let us consider a Standard BF of $m$ memory-positions, $n$ elements, $f$ probability of false positives, $k$ hash functions and $p$ probability that a bit is still 0 after inserting all $n$. Let us consider also $d$ Individual IBFs of $m_{ind_x}$ memory-positions, $n_{ind_x}$ elements, $f_{ind_x}$ probability of false positives, $k_{ind_x}$ hash functions and $p_{ind_x}$ probability that a bit is still 0 after inserting all $n_{ind_x}$ where $x$ indicates the level number of an Individual IBF.

\begin{theorem}
\label{reduce_m_maintain_f}
We state that an I(FIB)F may save memory-positions if the overall probability of false positives remains the same.
\end{theorem}

\begin{proof}
When hashing a name field to an Individual IBF more than once, $f_{ind_x}$ may be maintained: $f_{ind_x}\approx(1-p_{ind_x})^{k_{ind_x}}$, if $k_{ind_x}$ and $p_{ind_x}$ remain unchanged. As a consequence, the relation $\frac{m}{n}$ is also maintained because $k_{ind_x}=-\frac{m_{ind_x}}{n_{ind_x}}\ln p_{ind_x}$. Table \ref{table:standard vs iterated fixing p} shows that $n=n_{ind_x}$, however if the elements follow a hierarchical structure and the IBFs are hashed field by field then $n_{ind_x} < n$. As a result, $m_{ind_x}$ must be reduced so that IBFs save memory-positions when compared to Standard BFs: $\sum_{i=1}^{x}m_{ind_x} < m$. Therefore, we state that memory-positions in an I(FIB)F may be saved for maintaining $f_{ind_x}$ so that $f_i$ is also maintained.
\end{proof}

\begin{theorem}
\label{reduce_f_maintain_m}
We state that an I(FIB)F may reduce the overall probability of false positives if the overall memory-positions remain the same.
\end{theorem}

\begin{proof}
Table \ref{table:standard vs iterated fixing p} shows that $n=n_{ind_x}$, however if the elements follow a hierarchical structure and the IBFs are hashed field by field then $n_{ind_x} < n$. If $p_{ind_x}$ and $m_{ind_x}$ are maintained then $k_{ind_x}$ is increased: $k_{ind_x}=-\frac{m_{ind_x}}{n_{ind_x}}\ln p_{ind_x}$ and $m_i$ remains the same. As a result, $f_{ind_x}$ is reduced: $f_{ind_x}\approx(1-p_{ind_x})^{k_{ind_x}}$. Consequently, $f_i$ is also reduced. Then, IBFs reduce the overall probability of false positives when compared to Standard BFs: $f_i<f$. 
\end{proof}

\section{Routing protocol}
\label{sec:routing}


At the time of writing, NLSR \cite{hoque2013nlsr} is the default routing protocol used in NDN. NLSR is a link state protocol extended from classic OSPF which is widely used in intra-network routing. In the following, we first briefly recap the mechanisms of NLSR, then we show how to modify NSLR to accommodate our proposed solution. Herein, it is worth emphasizing that even though the corresponding modifications are necessary, such engineering efforts are trivial in practice.

NLSR propagates two types of Link-State Advertisement LSA. The first type (i.e., Adjacency LSA) is used to advertise a router's link state information to its directly connected neighbours. On the other hand, the second type (i.e., Prefix LSA) is used to advertise the name prefixes registered with the current node. Note that NLSR does not bundle multiple prefixes in one LSA. Instead, each prefix must be advertised separately (due to the high traffic cost by using hierarchical naming). Both LSA are properly wrapped into NDN $Interest$ and $Data$ packets and the carried information is stored in a Link State Database (LSDB) at each router. Whenever there are any changes in link states or registered prefixes, the changes will be advertised with the corresponding type of LSA and the LSDB will be synchronised as well in a hop-by-hop fashion using CCNx sync and repo protocols \cite{CCNx}.

With the information stored in LSDB, a node can first construct a weighted graph for the network it resides in. The weight on a link represents the cost of data transmission. By running Dijkstra's shortest path algorithm, the node further calculates the path to every known prefix in order to construct a forwarding table (i.e. FIB).

In our proposed solution, since hashes have replaced hierarchical names, Prefix LSA will carry these hashes instead of plain text names accordingly. The hashes can be either sent out separately or bundled in one LSA (without changing original NLSR semantics) thanks to its compact format. Whenever a Prefix LSA goes through a router it is necessary to introduce mechanisms to handle the iterated hashes. Meanwhile, Adjacency LSA remains the same as in original NLSR.

When constructing a forwarding table, a node first calculates the shortest path to every known node as before. Then the router simply ``OR'' the hashes of the registered content at a destination with the associated I(FIB)F with the corresponding link (i.e., next-hop link leading to the destination). Comparing to the original algorithm of building FIB in NLSR, our adapted version (see Section \ref{sec:forwarding}) uses a simple bit-wise ``OR'' operation rather than parsing prefixes and combining them into a single FIB, which obviously leads to much lower computation complexity.

\section{Related Work}
\label{sec:related}

Previous work on content management deals with lookup solutions based on BFs. In TB$^2$F (Tree-Bitmap and Bloom Filter) \cite{quan2014tb2f} a tree structure to save content in Content-Centric Networking is defined. Top leaves follow a T-segment Tree while bottom leaves require counting BFs for content storage. The solution proposed shows that if the structure is well-designed it provides good scalability.


Furthermore, some name lookup techniques are specifically designed for NDN. In Name Lookup engine with Adaptive Prefix Bloom filter (NLAPB) \cite{quan2014scalable} name prefixes are divided in B-prefixes and T-suffixes. Standard BFs match B-prefixes while a small-scale trie is used for T-suffixes. The division is based on the popularity of names to speed up the lookup.



In NameFilter \cite{wang2013namefilter} the lookup of names is achieved using two-stage BFs. The first BFs save name prefixes based on their lengths. Then, each next-hop port is represented by a BF. Each name prefix is associated to the appropriated BFs depending on their associated ports.



Thereafter NameFilter, a new technique to speed up name lookup in NDN \cite{wang2014fast} has been defined. First of all, it requires to compute the distribution of name prefixes with the aim of reducing the time of matching longest prefixes. Afterwards, perfect hash tables store the signature of prefixes.


All these strategies propose new methods to substitute the FIB. The main difference with our solution is that we require a FIB for each output interface of a node and we directly transmit the iterated hashes required for the lookup in the $Name$ of an $Interest$. As a consequence, the $Name$ received is checked against all possible output interfaces in a straightforward manner without requiring intermediate stages. Therefore, the complexity required is very low. Finally, it is not necessary to keep next-hops in the FIB but a routing protocol as NLSR to find the appropriate output interfaces for each $Interest$.


\section{Conclusion}
\label{sec:conclusion}
In this research, we propose the construction of a Forwarding Information Base based on Iterated Bloom Filters I(FIB)F for Name Data Networks (NDN). We focus our efforts on maximizing the efficiency of the design. This is due to the fact that many constrained devices in terms of CPU, memory, bandwidth and battery are expected in the Future Internet.

First of all, we study the impact of substituting hierarchical names on $Interests$ by iterated hashes. We conclude that our strategy reduces the overhead of packets. Additionally, the complexity of the forwarding strategy compared to current solutions is also reduced. One of the advantages of our solution is that we require an I(FIB)F per output interface instead of a single FIB that defines the next-hop. As a consequence, $Names$ are directly checked against all possible interfaces without intermediate steps. Additionally, we present how to integrate the proposed solution with existing NDN protocol stack with minimal efforts.

Furthermore, we determine that a design based on Bloom filters reduces the routing memory. When comparing I(FIB)F with a Standard BF we state that our design may reduce the overall memory or the probability of false positives. Moreover, we evaluate different estimation methods of content names for an accurate design. Our results show that an estimation method is needed per field of the hierarchical structure.

To sum up, we conclude that I(FIB)F for Name Data Networks is a highly efficient solution for the Future Internet.



\section*{Acknowledgment}
This research has been financially supported by the Swiss National Science Foundation with an Early Postdoc.Mobility Fellowship under grant agreement no. P2GEP2\_168977.

\bibliographystyle{IEEEtran}
\bibliography{ref}

\begin{thebibliography}{10}
\providecommand{\url}[1]{#1}
\csname url@samestyle\endcsname
\providecommand{\newblock}{\relax}
\providecommand{\bibinfo}[2]{#2}
\providecommand{\BIBentrySTDinterwordspacing}{\spaceskip=0pt\relax}
\providecommand{\BIBentryALTinterwordstretchfactor}{4}
\providecommand{\BIBentryALTinterwordspacing}{\spaceskip=\fontdimen2\font plus
\BIBentryALTinterwordstretchfactor\fontdimen3\font minus
  \fontdimen4\font\relax}
\providecommand{\BIBforeignlanguage}[2]{{%
\expandafter\ifx\csname l@#1\endcsname\relax
\typeout{** WARNING: IEEEtran.bst: No hyphenation pattern has been}%
\typeout{** loaded for the language `#1'. Using the pattern for}%
\typeout{** the default language instead.}%
\else
\language=\csname l@#1\endcsname
\fi
#2}}
\providecommand{\BIBdecl}{\relax}
\BIBdecl

\bibitem{kutscher2016information}
D.~Kutscher, S.~Eum, K.~Pentikousis, I.~Psaras, D.~Corujo, D.~Saucez,
  T.~Schmidt, and M.~Waehlisch, ``Information-centric networking (icn) research
  challenges,'' Tech. Rep., 2016.

\bibitem{zhang2014named}
L.~Zhang, A.~Afanasyev, J.~Burke, V.~Jacobson, P.~Crowley, C.~Papadopoulos,
  L.~Wang, B.~Zhang \emph{et~al.}, ``Named data networking,'' \emph{ACM SIGCOMM
  Computer Communication Review}, vol.~44, no.~3, pp. 66--73, 2014.

\bibitem{katz2014introduction}
J.~Katz and Y.~Lindell, \emph{Introduction to modern cryptography}.\hskip 1em
  plus 0.5em minus 0.4em\relax CRC press, 2014.

\bibitem{munoz2015fragmented}
C.~Mu{\~n}oz and P.~Leone, ``Fragmented-iterated bloom filters for routing in
  distributed event-based sensor networks,'' pp. 248--261, 2015.

\bibitem{Standard}
S.~Tarkoma, C.~Rothenberg, and E.~Lagerspetz, ``Theory and practice of bloom
  filters for distributed systems,'' \emph{Communications Surveys Tutorials,
  IEEE}, vol.~14, no.~1, pp. 131--155, First 2012.

\bibitem{ThesisCFIBFs}
\BIBentryALTinterwordspacing
C.~Mu{\~n}oz, ``A distributed event-based system based on fragmented-iterated
  bloom filters,'' 01/19 2016. [Online]. Available:
  \url{http://nbn-resolving.de/urn:nbn:ch:unige-860190}
\BIBentrySTDinterwordspacing

\bibitem{hoque2013nlsr}
A.~Hoque, S.~O. Amin, A.~Alyyan, B.~Zhang, L.~Zhang, and L.~Wang, ``Nlsr:
  named-data link state routing protocol,'' in \emph{Proceedings of the 3rd ACM
  SIGCOMM workshop on Information-centric networking}.\hskip 1em plus 0.5em
  minus 0.4em\relax ACM, 2013, pp. 15--20.

\bibitem{bolstad2013introduction}
W.~M. Bolstad, \emph{Introduction to Bayesian statistics}.\hskip 1em plus 0.5em
  minus 0.4em\relax John Wiley \& Sons, 2013.

\bibitem{guo2010dynamic}
D.~Guo, J.~Wu, H.~Chen, Y.~Yuan, and X.~Luo, ``The dynamic bloom filters,''
  \emph{IEEE Transactions on Knowledge and Data Engineering}, vol.~22, no.~1,
  pp. 120--133, 2010.

\bibitem{bonomi2006improved}
F.~Bonomi, M.~Mitzenmacher, R.~Panigrahy, S.~Singh, and G.~Varghese, ``An
  improved construction for counting bloom filters,'' in \emph{European
  Symposium on Algorithms}.\hskip 1em plus 0.5em minus 0.4em\relax Springer,
  2006, pp. 684--695.

\bibitem{shannon1951prediction}
C.~E. Shannon, ``Prediction and entropy of printed english,'' \emph{Bell system
  technical journal}, vol.~30, no.~1, pp. 50--64, 1951.

\bibitem{murphy2007conjugate}
K.~P. Murphy, ``Conjugate bayesian analysis of the gaussian distribution,''
  \emph{def}, vol.~1, no. 2$\sigma$2, p.~16, 2007.

\bibitem{adamic2002zipf}
L.~A. Adamic and B.~A. Huberman, ``Zipf’s law and the internet,''
  \emph{Glottometrics}, vol.~3, no.~1, pp. 143--150, 2002.

\bibitem{CCNx}
\BIBentryALTinterwordspacing
``Parc, ccnx open source platform.'' 2016. [Online]. Available:
  \url{http://www.ccnx.org}
\BIBentrySTDinterwordspacing

\bibitem{quan2014tb2f}
W.~Quan, C.~Xu, A.~V. Vasilakos, J.~Guan, H.~Zhang, and L.~A. Grieco, ``Tb2f:
  Tree-bitmap and bloom-filter for a scalable and efficient name lookup in
  content-centric networking,'' in \emph{Networking Conference, 2014
  IFIP}.\hskip 1em plus 0.5em minus 0.4em\relax IEEE, 2014, pp. 1--9.

\bibitem{quan2014scalable}
W.~Quan, C.~Xu, J.~Guan, H.~Zhang, and L.~A. Grieco, ``Scalable name lookup
  with adaptive prefix bloom filter for named data networking,'' \emph{IEEE
  Communications Letters}, vol.~18, no.~1, pp. 102--105, 2014.

\bibitem{wang2013namefilter}
Y.~Wang, T.~Pan, Z.~Mi, H.~Dai, X.~Guo, T.~Zhang, B.~Liu, and Q.~Dong,
  ``Namefilter: Achieving fast name lookup with low memory cost via applying
  two-stage bloom filters,'' in \emph{INFOCOM, 2013 Proceedings IEEE}.\hskip
  1em plus 0.5em minus 0.4em\relax IEEE, 2013, pp. 95--99.

\bibitem{wang2014fast}
Y.~Wang, B.~Xu, D.~Tai, J.~Lu, T.~Zhang, H.~Dai, B.~Zhang, and B.~Liu, ``Fast
  name lookup for named data networking,'' in \emph{2014 IEEE 22nd
  International Symposium of Quality of Service (IWQoS)}.\hskip 1em plus 0.5em
  minus 0.4em\relax IEEE, 2014, pp. 198--207.

\end{thebibliography}

\end{document}